\documentclass[review]{elsarticle}
\usepackage{amsfonts}
\usepackage{lineno,hyperref}

\usepackage{amsmath,amssymb,amsfonts}
\usepackage{algorithmic}
\usepackage{graphicx}
\usepackage{textcomp}

\usepackage{graphicx}
\usepackage{subfigure}
\usepackage{epsfig}
\modulolinenumbers[5]
\journal{Systems \& Control Letters}
\newtheorem{theorem}{\textbf{Theorem}}

\newtheorem{remark}{\textbf{Remark}}
\newtheorem{assumption}{\textbf{Assumption}}
\input{tcilatex}
\begin{document}

\begin{frontmatter}

\title{Multi-discontinuous Functional based Sliding Mode Cascade Observer for Estimation and Closed-loop Compensation Controller}

\author[Firstaddress]{Yiyong Sun}
%\ead{sunyy@bit.edu.cn,yiyonghit@gmail.com}
%
\author[Secondaddress]{Zhang Chen}

\author[Firstaddress]{Guang Zhai\corref{mycorrespondingauthor}}
\ead{gzhai@bit.edu.cn}
\cortext[mycorrespondingauthor]{Corresponding author}

\author[Secondaddress]{Bin Liang}

\address[Firstaddress]{School of Aerospace Engineering, Beijing Institute of Technology, Beijing, 100081, China}
\address[Secondaddress]{Department of Automation, Tsinghua University, Beijing, 100084, China}
%\address[mysecondaryaddress]{the Chair
%of Automatic Control Engineering (LSR), Technische Universit\"{a}t M\"{u}%
%nchen, Theresienstr. 90, 80333 M\"{u}nchen, Germany}

\begin{abstract}
The sliding mode observer is a useful method for estimating the system state and the unknown disturbance.
However, the traditional single-layer observer might still suffer from high pulse when the output measurement is mixed with noise.
To improve the estimation quality, a new cascade sliding mode observer containing multiple discontinuous functions is proposed in this letter.
The proposed observer consists of two layers: the first layer is a traditional sliding mode observer, and the second layer is a cascade observer.
The measurement noise issue is considered in the source system model.
An alternative method how to design the observer gains of the two layers, together with how to examine the effectiveness of the compensator based closed-loop system, are offered.
A numerical example is provided to demonstrate the effectiveness of the proposed method.
The observation structure proposed in this letter not only smooths the estimated state but also reduces the control consumption.
\end{abstract}


\end{frontmatter}

%\linenumbers

%%%%%%%%%%%%%%%%%%%%%%%%%%%%%%%%%%%%%%%%%%%%%%%%%%%%%%%%%%%%%%%%%%%%%%%%%%%%
\section{Introduction} \label{sec:introduction}

Together with the development of variable structure control and nonlinear discontinuous control theory, the sliding mode observer (SMO) technique draws attention from researchers and engineers due to its adaptivity, disturbance estimation and compensation ability for linear and nonlinear systems, especially after the new century \cite{drakunov1995sliding, spurgeon2008sliding, shtessel2014sliding, zhang2023disturbance}.
The SMO technique was introduced as early as in 1980s \cite{slotine1987sliding,walcott1987state} and then applied for fault detection and isolation \cite{edwards2000sliding,edwards2000slidingEJC}, actuator and sensor fault reconstruction and detection considering system uncertainty \cite{tan2003sliding}.
In the last decade, the SMO technique has been further developed and broadly used.
For example, the sliding mode observer method is applied for the predictive current control for permanent magnet synchronous motor drive systems in \cite{li2023improved},
using the descriptor augment remodelling method, the SMO is expanded to the fault tolerant control of nonlinear systems  \cite{gao2019fault}, fault reconstruction, sensor and actuator fault estimation of stochastic switching and hybrid systems \cite{liu2018fault,yin2017descriptor,yang2018descriptor}.

Most existing research on SMO-based feed-forward compensation controllers uses a single observation layer, which results in noisy disturbance and system state estimation, and causes actuator vibration due to measurement noises and switching function.
Various methods have been proposed to improve the traditional SMO and avoid actuator vibration.
These methods include using high order SMO \cite{shtessel2014sliding}, introducing the optimized switching function \cite{kyslan2022comparative}, or combining SMO with other filters \cite{sun2017coupled}.
The cascade high gain observer technique, which was developed in recent years \cite{khalil2017cascade}, shows potential against peaking signals and model uncertainties.
The cascade sliding mode observer is employed on the torque-sensorless control of permanent-magnet synchronous machines \cite{zhao2017accurate}, but the measurement noise is not fully considered.

In this letter, a two-layer SMO containing SMO and cascade observer, is provided to further smooth the estimated state and disturbance.
The source system model considered in this letter takes into account both measurement noise and lumped disturbance.
The original system model is then transferred into a new descriptor one via the system state augmentation technique.
The existing single layer SMO and compensation controller based on it are reviewed.
Then, the two-layer observer, i.e. SMO-CO, whose first layer is the traditional SMO and the second layer is the cascade one, is proposed.
An alternative method for selecting the observe gains of the SMO-CO, and the sufficient condition for examining the effectiveness of the closed-loop system, are presented.

The main contribution of this letter is as follows:
\begin{itemize}
  \item It proposes a new SMO-CO based observer scheme, which can further smooth the estimated disturbance and system state.
         Unlike existing research on SMO, which only has one discontinuous function, the observer proposed here has multiple ones.
  \item It presents an alternative method for designing the gains of the two-layer observer.
        It also provides a sufficient condition for evaluating the closed-loop system with an observer based compensation controller.
  \item It shows that, compared with the conventional single layer SMO, the SMO-CO scheme proposed in this letter has lower observation error and less control consumption.
\end{itemize}

The rest of this letter is organized as follows.
In Section. \ref{sec:Preliminary}, the conventional SMO based compensator is firstly introduced.
The two-layer SMO-CO control scheme, methods on designing the two observer gains, designing the discontinuous functions, and the sufficient condition on examining the closed-loop system are presented in Section. \ref{sec:MainResult}.
To examine the validity of the proposed SMO-CO based compensation control scheme, a numerical example is offered in Section.\ref{sec:Example}.
Section.\ref{sec:Conclusion} concludes the work of the whole letter.

%%%%%%%%%%%%%%%%%%%%%%%%%%%%%%%%%%%%%%%%%%%%%%%%%%%%%%%%%%%%%%%%%%%%%%%%%%%%

\section{Preliminary} \label{sec:Preliminary}
\subsection{System Description}
Consider the following system  \eqref{Equ_SMO_Original_System}  with unknown lumped system disturbance $d(t)$ and measurement noise $\omega(t)$
\begin{equation}
\left\{
\begin{array}{rcl}
\dot{x}(t) &=& A x(t) + Bu(t)+B_f d(t)  \\
y(t) &=& Cx(t) + C_\omega \omega(t)
\end{array}
 \right.
 \label{Equ_SMO_Original_System}
\end{equation}
where $A\in \mathcal{R}^{n\times n}$, $B\in \mathcal{R}^{n\times m}$, $B_f = B \Lambda\in \mathcal{R}^{n\times m} $, $C\in \mathcal{R}^{p\times n}$  are the system parameters.
$\Lambda \in \mathcal{R}^{m\times m}$ is nonsingular.
$x(t)\in \mathcal{R}^{n}$, $u(t)\in \mathcal{R}^{m}$, $d(t)\in \mathcal{R}^{m}$, and $y(t)\in \mathcal{R}^{p}$ indicate the system state, control input, lumped disturbance and measurement output vectors respectively.
$C_\omega  \in \mathcal{R}^{p\times p} \geq 0$ is the coefficient matrix for the standard unit Gaussian noise $\omega(t)\sim (0, 1)\in \mathcal{R}^{p}$.
$\omega(t)\in [-\overline{\omega},\overline{\omega} ]$, $\overline{\omega}$ is the maximal amplitude and set to be $1$ for unit Gaussian noise signal.
To be concise, in the following description, the vectors and time-constant matrices are expressed into brief forms without time $t$. Pairs $(A,B)$ and $(A,C)$ are controllable and observable.

\textbf{Objective:}
The objective of this paper is to design a new observer scheme for estimating the unknown disturbance $d(t)$ and system state, and smoothing the system state with measured noise.

\begin{assumption} \label{Asmp_Upper_Bdr}
The lumped disturb $d(t)$ in \eqref{Equ_SMO_Original_System} might consist of the un-modelled system uncertainties, unknown external perturbation like friction, artificial interrupt. It is assumed to be amplitude limited and  Lipschitz \cite{gao2019fault,yu2017fault}, which means
\begin{equation}
\| d(t) \| \leq \overline{\mathbf{d}},  \| \dot{d} (t) \| \leq \overline{\mathbf{h}}
\label{Equ_Assumption_1_bdr}
\end{equation}
where $\overline{\mathbf{d}}$ and $\overline{\mathbf{h}}$ are the upper boundaries of the lumped disturbance $d(t)$ and its derivative $\dot{d}(t)$.
\end{assumption}

\subsection{Descriptor Augment Model}
Using the descriptor augment technique in \cite{gao2019fault,yu2017fault,yin2017descriptor,gao2007actuator,pang2020design,yang2018descriptor,chen2019fault}, the system \eqref{Equ_SMO_Original_System} is augmented into equal form below
\begin{equation}
\left\{
\begin{array}{rcl}
\bar{E} \dot{\bar{x}} &=& \bar{A} \bar{x} + \bar{B}u + \bar{B}_f \bar{d} \\
y &=& \bar{C}\bar{x} + C_\omega  \omega
\end{array}
 \right.
 \label{Equ_SMO_Augment_System}
\end{equation}
where $ \bar{x}(t)  =  [x^T(t) \ d^T(t)]^T$,$\bar{C}   =  \left[ C \ 0_{p\times m } \right] , \Phi \geq 0$,
\begin{eqnarray*}
&\bar{E}& = \left[ \begin{array}{cc} I & B_f\Phi^{-1} \\ 0 & I \end{array}  \right],
\bar{A} = \left[ \begin{array}{cc} A & 0 \\ 0 & -\Phi \end{array}  \right], \\
&\bar{B}&  =  \left[ \begin{array}{c} B \\ 0 \end{array} \right],
\bar{B}_f = \left[ \begin{array}{c} B_f \\ \Phi \end{array} \right],
\bar{d}  = \left( d+\Phi^{-1} \dot{d}  \right).
\end{eqnarray*}

\begin{remark}
The matrix $\bar{E} $ is nonsingular that, one can multiply the left and right sides of the differential equation in \eqref{Equ_SMO_Augment_System} to obtain a normal dynamics.
But it's necessary on designing the SMO in the following, one can design observer and analyze its effective via the descriptor form in \eqref{Equ_SMO_Augment_System} directly.
And, such a treatment is potential when the sensor fault is considered in the system.
\end{remark}
\begin{remark}
The matrix $\Phi$ would be selected accordingly.
The effectiveness of $\Phi$ is it can modulate the gain $\bar{B}_f$.
When $\Phi$ is set to be identity matrix, $\bar{E}$ becoming to be identity, \eqref{Equ_SMO_Augment_System} is in the form as in \cite{yin2017descriptor,gao2019fault,yang2020neural} without considering sensor faults.
\end{remark}

\subsection{Augmented Sliding Mode Observer}
With \eqref{Equ_SMO_Augment_System}, one has the augmented state based sliding mode observer below
\begin{equation}
\left\{
\begin{array}{l}
\bar{E} \dot{\hat{\bar{x}}}  =  \bar{A} \hat{\bar{x}} + \bar{B}u + \bar{L}(y-\bar{C} \hat{\bar{x}})+\bar{B}_f u_{s1}(t)- \bar{L} C_{\omega} u_{s2}(t) \\
\hat{y}  =  \bar{C}\hat{\bar{x}}
\end{array}
\right.
\label{Equ_Aug_SMO}
\end{equation}

With the augmented descriptor system \eqref{Equ_SMO_Augment_System} and the observer system \eqref{Equ_Aug_SMO}, one defines the augmented observer error $\bar{e}(t)= \bar{x}(t) -\hat{\bar{x}}(t) = \left[e^T(t) e_d^T(t)\right]^T$, where $e(t)= x(t) -\hat{x}(t)$ is the system state observation error, and $e_d(t) = d(t)-\hat{d}(t)$ indicates the disturbance observation error. The observation error dynamics is
\begin{equation}
\left\{
\begin{array}{rcl}
\bar{E} \dot{\bar{e}}  &=& \left(\bar{A}  - \bar{L}\bar{C}\right)\bar{e} +\bar{L} {C}_{\omega} (u_{s2}(t)-\omega) + \bar{B}_f(\bar{d}  - u_{s1}(t)) \\
e_{y}  &=& \bar{C}\bar{e} + C_{\omega}  \omega(t)
\end{array}
\right.
\label{Equ_ObserverError}
\end{equation}
If one can design the gain $\bar{L}$, $u_{s1}(t)$ and $u_{s2}(t)$ that $\bar{e}(t)$ approaches to be zero, i.e. $\lim_{t\rightarrow \infty} \bar{e}(t)\rightarrow 0$ and \eqref{Equ_ObserverError} is stable, the effectiveness of \eqref{Equ_Aug_SMO} can be guaranteed.

The discontinuous functions $u_{s1}(t)$ and $u_{s2}(t)$ are designed
\begin{equation}
\begin{array}{rcl}
u_{s1}(t) &=&  (\overline{\mathbf{d}}+ \overline{\mathbf{h}} \Phi^{-1}+\eta)  \textbf{sgn}\left(s_1(t) \right) \\
u_{s2}(t) &=&  -\overline{\mathbf{\omega}}  \textbf{sgn}\left(s_2(t) \right)
\end{array}
\label{Equ_Discontinus_Function}
\end{equation}
where
$\overline{\mathbf{d}}$ and $ \overline{\mathbf{h}}$ are assumed as in \eqref{Equ_Assumption_1_bdr} of \textbf{Assumption \ref{Asmp_Upper_Bdr}}.
The parameter $\eta>0 $ is to be selected properly. The sign functions of $u_{s1}(t)$ and $u_{s2}(t)$ are selected as
\begin{equation}
s_1(t) = H_1 \bar{C} \bar{e}(t), s_2(t) = H_2 \bar{C} \bar{e}(t).
\end{equation}
Matrices $H_1$ and $H_2$ are selected to satisfy below constraints
\begin{equation}
(H_1\bar{C})^T = \bar{P}\bar{E}^{-1} \bar{B}_f, (H_2\bar{C})^T = \bar{P}\bar{E}^{-1}\bar{L}C_{\omega}
\label{Equ_HC_Constraint}
\end{equation}
where $\bar{P}$ and $\bar{L}$ are the non-negative matrix determined in the following theorem and the observer gain. How to select the gains $H_1$ and $H_2$ are presented in Section.\ref{Subsec:CAS}.

\begin{theorem} \label{Theo_SMO}
With the discontinuous functions $u_{s1}(t)$ and $u_{s2}(t)$, the observer error system \eqref{Equ_ObserverError} is stable, i.e. \eqref{Equ_Aug_SMO} is effective, when there exist positive matrices $\bar{P} $, and matrix $\bar{N}$ with appropriate dimension that
\begin{equation}
0 > \bar{P} \bar{E}^{-1} \bar{A}+ \bar{A}^T\bar{E}^{-T}\bar{P}- \bar{N}\bar{C}-\bar{C}^T\bar{N}^T
\label{Equ_H2_Rqr}
\end{equation}
and the observe gain $\bar{L}$ is designed as $\bar{L} =  \bar{E}\bar{P}^{-1} \bar{N} $.
\end{theorem}

\begin{proof}
Define the Lyapunov functional below
\begin{equation}
V(t) = \bar{e}^T(t) \bar{P} \bar{e}(t)
\end{equation}
whose derivative is
\begin{equation*}
\begin{array}{rcl}
\dot{V}(t)
&=& 2 \bar{e}^T(t) \bar{P} \bar{E}^{-1} \left(\bar{A}  - \bar{L}\bar{C}\right)\bar{e}(t) \\
&& +2 \bar{e}^T(t) \bar{P} \bar{E}^{-1}\bar{L} {C}_{\omega} (u_{s2}(t)-\omega) \\
&& +2 \bar{e}^T(t) \bar{P} \bar{E}^{-1}\bar{B}_f\left(\bar{d}  - u_{s1}(t)\right)
\end{array}
\end{equation*}
If one defines matrix $\bar{N}$ that $\bar{N} =\bar{P}\bar{E}^{-1}\bar{L}$, and the given constraints in \eqref{Equ_HC_Constraint}, one has $(H_1\bar{C})^T = \bar{P}\bar{E}\bar{B}_f$ and $(H_2\bar{C})^T = \bar{P}\bar{E}^{-1}\bar{L}C_{\omega} = \bar{N}C_{\omega}$, therefore
\begin{eqnarray*}
&& \bar{e}^T(t) \bar{P}  \bar{E}^{-1} \bar{B}_f(\bar{d} (t)- u_{s1}(t))\\
&=& \bar{e}^T(t) \bar{P}  \bar{E}^{-1}  \bar{B}_f \left(\bar{d}(t)- (\overline{\mathbf{d}}+ \overline{\mathbf{h}} \Phi^{-1}+\eta)  \textbf{sgn}\left(s_1(t) \right)\right) \\
&= & s_1^{T}(t) \left(\bar{d}(t)- (\overline{\mathbf{d}}+ \overline{\mathbf{h}} \Phi^{-1}+\eta)  \textbf{sgn}\left(s_1(t) \right)\right) \\
%&\leq & |s^{T}(t)| \left(|\bar{d}(t)|- (\overline{\mathbf{d}}+ \overline{\mathbf{h}} \Phi^{-1}+\eta) \right) \\
&\leq &  |s_1^{T}(t)| \left(|{d}(t)|+|\Phi^{-1}\dot{d}(t)|- (\overline{\mathbf{d}}+ \overline{\mathbf{h}} \Phi^{-1}+\eta) \right)\\
&\leq & - \eta |s_1^{T}(t)|\\
&\leq & 0
\end{eqnarray*}
and the gaussian noise $\omega \in [-\bar{\omega}, \bar{\omega}]$
\begin{eqnarray*}
 &&   \bar{e}^T(t) \bar{P} \bar{E}^{-1}\bar{L} {C}_{\omega} (u_{s2}(t)-\omega) \\
 &=&  \bar{e}^T(t) \bar{P} \bar{E}^{-1}\bar{L} {C}_{\omega} (-\bar{\omega}\textbf{sgn}\left(s_2(t) \right)-\omega) \\
 &=&  s^T_2(t)(-\bar{\omega}\textbf{sgn}\left(s_2(t) \right)-\omega) \\
 &\leq& (-\bar{\omega}-\omega) |s_2^T| \\
 &\leq& 0
\end{eqnarray*}
And, one can further verify that, if the inequality \eqref{Equ_H2_Rqr} holds
\begin{equation*}
\begin{array}{rcl}
\dot{V}(t) &\leq&  0
\end{array}
\end{equation*}
The system \eqref{Equ_ObserverError} is stable.
\end{proof}

\subsection{Constraint Approximation Solution} \label{Subsec:CAS}
The constraints \eqref{Equ_HC_Constraint} and the discontinuous functions \eqref{Equ_Discontinus_Function} in \eqref{Equ_Aug_SMO} are vital.
The gains $H_1$ and $H_2$ can be computed approximately via the following method.

For the equation of \eqref{Equ_HC_Constraint}, one can obtain the gain $H_i$, for $i=1,2$, by the following form approximatively
\begin{equation}
\text{Trace} \left[ \Xi  \Xi ^T\right] =0
\end{equation}
which means there exists a positive scalar $\mu$ that
\begin{equation}
\Xi \Xi^T \leq \mu I
\end{equation}
with the Schur complement theory, one has that
\begin{equation}
\left[
\begin{array}{cc}
-\mu I & \Xi \\
\ast & -I
\end{array}
\right] \leq 0
\end{equation}
where $\Xi$ can be $(H_1\bar{C})^T - \bar{P}\bar{E}^{-1} \bar{B}_f$ or $(H_2\bar{C})^T - \bar{P}\bar{E}^{-1}\bar{L}C_{\omega}$ and $\bar{P}$ is the solution value in \textbf{Theorem \ref{Theo_SMO}}.

\subsection{SMO based Compensator}
\begin{figure}[h]
\centering
	\includegraphics[ width=0.8\linewidth]{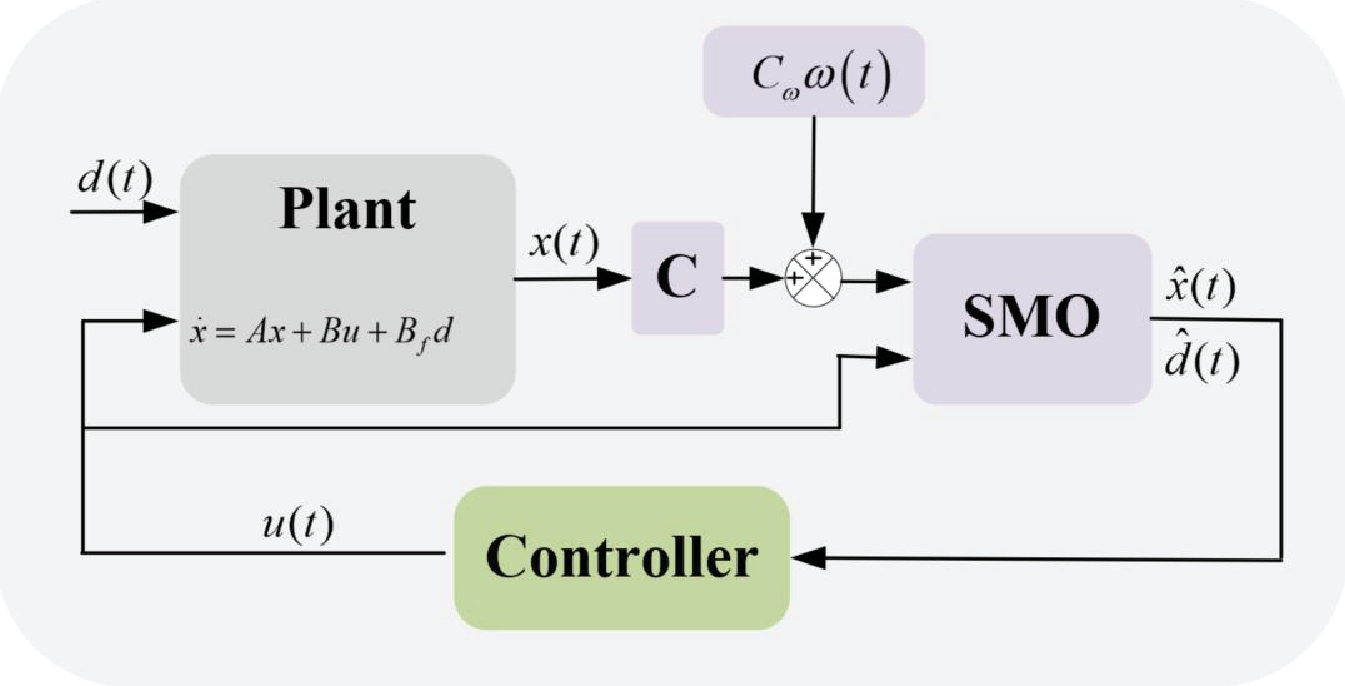}
	\caption{Traditional SMO based Controller}
    \label{Fig_PSMO_Controller}
\end{figure}

Depicted as in Fig.\ref{Fig_PSMO_Controller}, the SMO based control structure can be designed as follow
\begin{equation}
u(t) = \bar{K} \hat{\bar{x}}(t)
\label{Equ_SMO_Controller}
\end{equation}
where $\bar{K} = \left[K ~ -B^{\dagger}B_f \right]$, $B{\dagger}$ makes that $BB{\dagger}B_f=B_f$.
Then one can obtain the observer state based closed-loop system below
\begin{equation}
\left\{
\begin{array}{rcl}
\dot{x}                &=& (A+BK) x  + [-BK \ B_f] \bar{e}   \\
\bar{E} \dot{\bar{e}}  &=&  \left(\bar{A}  - \bar{L}\bar{C}\right)\bar{e} +\bar{L} {C}_{\omega} (u_{s2}(t)-\omega) \\
&&+ \bar{B}_f(\bar{d}  - u_{s1}(t))
\end{array}
\right.
\label{Equ_SMO_Closed_Loop}
\end{equation}
the effectiveness of the closed-loop system \eqref{Equ_SMO_Closed_Loop} can be examined via the sufficient condition in the following theorem.

\begin{theorem} \label{Theorem_SMO_Compensator}
The controller \eqref{Equ_SMO_Controller} and the observer \eqref{Equ_Aug_SMO} are effective, i.e. \eqref{Equ_SMO_Closed_Loop} is stable, when there exist positive matrices $\bar{P}$, $Q$ and matrix $\bar{N}$ with proper dimensions under the given controller gain $\bar{K}$ such that
\begin{equation}
\begin{array}{rcl}
\Xi &=& \left[ \begin{array}{cc} \Xi_{11} & \Xi_{12} \\ \Xi_{12}^T & \Xi_{22} \end{array}\right] <0
\end{array}
\label{Equ_ClsdLp_Xi}
\end{equation}
where
\begin{equation*}
\begin{array}{rcl}
\Xi_{11} &=& Q (A+BK)+  (A+BK)^T Q , \\
\Xi_{12} &=& Q [-BK \ B_f] ,  \\
\Xi_{22} &=& \bar{P} \bar{E}^{-1} \bar{A}+ \bar{A}^T\bar{E}^{-T}\bar{P}- \bar{N}\bar{C}-\bar{C}^T\bar{N}^T
\end{array}
\end{equation*}
and the observer gain $\bar{L}$ is designed as $\bar{L} =  \bar{E}\bar{P}^{-1} \bar{N} $.
\end{theorem}
\begin{proof}
Set Lyapunov equation below
\begin{equation}
V(t) = V_1(t)+ V_2(t)
\end{equation}
where
\begin{eqnarray*}
V_1(t)= x^T(t) Q x(t),  V_2(t)= \bar{e}^T (t)\bar{P} \bar{e}(t)
\end{eqnarray*}
Similar to \textbf{Theorem}.\ref{Theo_SMO}, the derivative of $V_2(t)$ is negative when there exist matrices $P$, $\bar{N}$ satisfy the inequality.
The derivative of $V_1(t)$ is
\begin{equation}
\begin{array}{rcl}
\dot{V}_1(t)
&=& 2x^T  Q(A+BK) x  + 2x^T  Q[-BK \ B_f] \bar{e}
\end{array}
\end{equation}
Then, if $\Xi<0$ is satisfied, one can determine $\dot{V}(t)= \dot{V}_1(t)+\dot{V}_2(t)<0$ holds, \eqref{Equ_SMO_Closed_Loop} is stable.
\end{proof}

\begin{remark}
The term $\bar{C}\bar{e}(t)$, is vital in designing the sign function $s_{i}(t)$ and discontinuous function $u_{si}(t)$ for $i=1,2 $, but cannot be calculated directly, especially when the measured output $y(t)$ is mixed with noise.
In this paper by the approximation technique, $\bar{C}\bar{e}(t)$ is computed as $\bar{C}\bar{e}(t)= y(t)- \bar{C}\hat{\bar{x}}(t)$.
\end{remark}

%%%%%%%%%%%%%%%%%%%%%%%%%%%%%%%%%%%%%%%%%%%%%%%%%%%%%%%%%%%%%%%%%%%%%%%%%%%%
%%%%%%%%%%%%%%%%%%%%%%%%%%%%%%%%%%%%%%%%%%%%%%%%%%%%%%%%%%%%%%%%%%%%%%%%%%%%
%%%%%%%%%%%%%%%%%%%%%%%%%%%%%%%%%%%%%%%%%%%%%%%%%%%%%%%%%%%%%%%%%%%%%%%%%%%%
%%%%%%%%%%%%%%%%%%%%%%%%%%%%%%%%%%%%%%%%%%%%%%%%%%%%%%%%%%%%%%%%%%%%%%%%%%%%
\section{Main Result} \label{sec:MainResult}
The estimated disturbance by the SMO method might be mixed with noise and random pulses, which can lead to actuator vibration and be unsuitable for subsequent use in engineering applications. In this section, a two-layer observation structure-based compensation controller is designed. It consists of the traditional SMO from the previous section and a cascade observer (CO), as depicted \ref{Fig_CSMO}.

\begin{figure}[h]
\centering
	\includegraphics[ width=0.8 \linewidth]{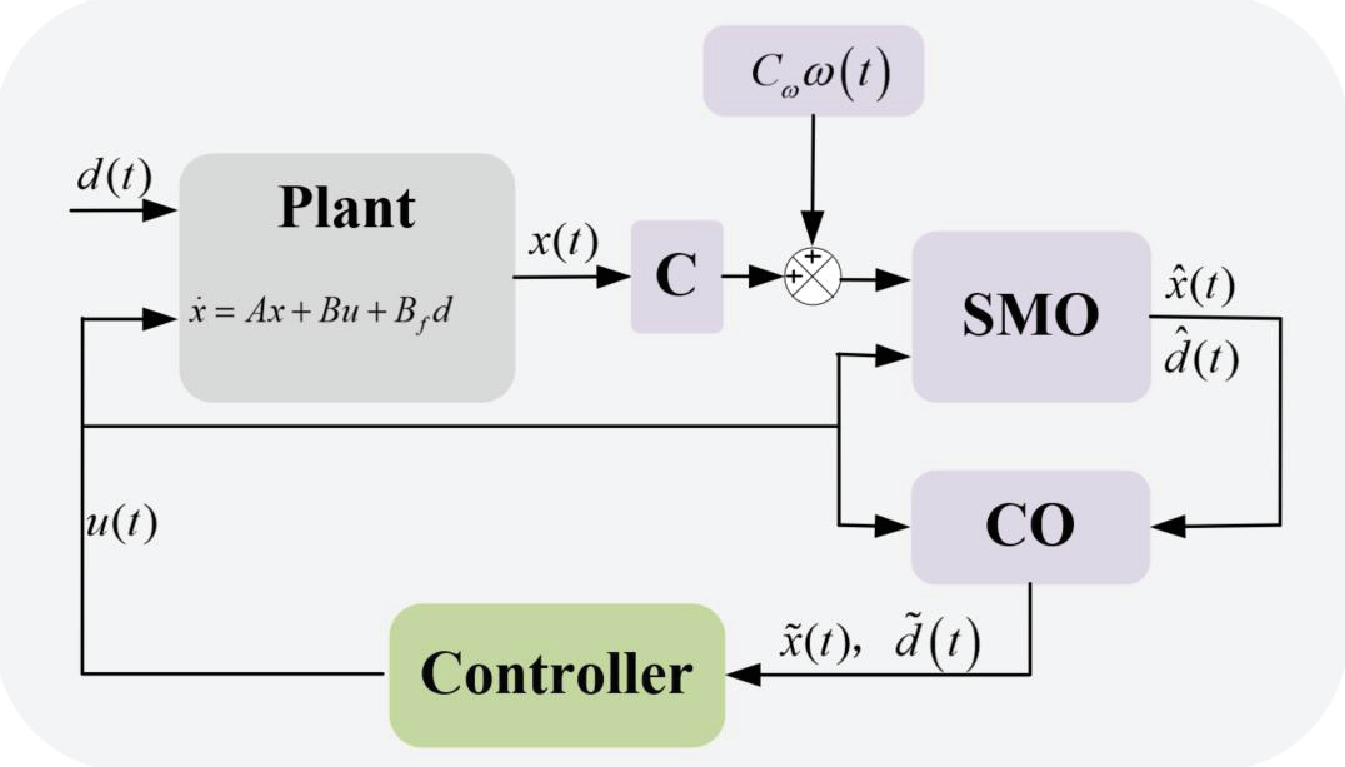}
	\caption{Cascade Observer}
    \label{Fig_CSMO}
\end{figure}

\subsection{Cascade Observer}
The cascade observer is designed to be
\begin{equation}
\begin{array}{rcl}
\bar{E} \dot{\tilde{\bar{x}}}  &=&  \bar{A} \tilde{\bar{x}} + \bar{B}{u} + \bar{\mathcal{L}}(\hat{\bar{x}}- \tilde{\bar{x}}) +\bar{B}_f u_{s1}(t)\\
                               & & - \bar{L} C_{\omega} u_{s2}(t)+\bar{L} C_{\omega} u_{s3}(t)
\end{array}
\label{Equ_Cascade_Observer}
\end{equation}
where $\tilde{\bar{x}}(t)= [\tilde{x}^T(t) $$ ~ \tilde{d}^T(t)]^T$ estimates $\hat{\bar{x}}(t)= $ $[\hat{x}^T(t) ~ $ $\hat{d}^T(t)]^T$.
$u_{s1}(t)$ and $u_{s2}(t)$ are the same as in \eqref{Equ_Discontinus_Function}, $u_{s3}(t)$ is to be designed in this section.
The observer gains $L$, which is also used in the SMO layer, and $\bar{\mathcal{L}}$ should be designed simultaneously to guarantee the effectiveness of the cascade observer state $\tilde{\bar{x}}(t)$.
%
%\begin{equation}
%\bar{E} \dot{\hat{\bar{x}}} = \bar{A} \hat{\bar{x}} + \bar{B}u + \bar{L}(y-\bar{C} \hat{\bar{x}})+\bar{B}_f u_{s3}(t)
%\end{equation}
%
%\begin{equation*}
%\begin{array}{rcl}
%\bar{E} \dot{\bar{e}}  &=& \left(\bar{A}  - \bar{L}\bar{C}\right)\bar{e} -\bar{L} {C}_{\omega} \omega + \bar{B}_f(\bar{d}  - u_{s3}(t))\\
%\bar{E} \dot{\tilde{\bar{\epsilon}}} &=& (\bar{A}-\bar{\mathcal{L}}) \tilde{\bar{\epsilon}}+\bar{L}\bar{C}\bar{e}(t)+\bar{L} C_{\omega} \omega
%\end{array}
%\end{equation*}

Define the cascade observer error $\bar{\epsilon}(t)$ below
\begin{equation*}
\bar{\epsilon}(t) = \hat{\bar{x}}(t) - \tilde{\bar{x}}(t)
\end{equation*}
one has that
\begin{equation}
\left\{
\begin{array}{l}
\bar{E} \dot{\bar{e}}  =  \left(\bar{A}  - \bar{L}\bar{C}\right)\bar{e}  +\bar{L} {C}_{\omega} (u_{s2}(t)-\omega)  + \bar{B}_f(\bar{d}  - u_{s1}(t)) \\
\bar{E} \dot{\bar{\epsilon}} = \left( \bar{A}  -\bar{\mathcal{L}}\right)\bar{\epsilon}+ \bar{L}\bar{C}\bar{e}+ \bar{L}C_{\omega}(\omega(t) -  u_{s3}(t)) \\
\end{array}
\right.
\label{Equ_CSMOErr}
\end{equation}
and the new discontinuous functional $u_{s3}(t)$ is selected to be
\begin{equation}
u_{s3}(t) = \overline{\omega} \textbf{sgn}(s_{3}(t))
\label{Equ_Discontinus_Function_S}
\end{equation}
where
\begin{equation}
s_3(t) = H_3 \bar{\epsilon}(t)
\end{equation}
matrix $\bar{P} $ is related with sufficient theorem given below that
\begin{equation}
\bar{P}  \bar{E}^{-1} \bar{L} C_{\omega} = \bar{N}_1 C_{\omega} = H_3^T
\label{Equ_HC_Constraint_S}
\end{equation}
This constraint is also computed via the method in subsection.\ref{Subsec:CAS}.
\begin{theorem} \label{Theorem_SOM_CO}
With the discontinuous functions $u_{si}(t)$, for $i=1,2,3$, the error dynamics \eqref{Equ_CSMOErr} is stable, i.e. the SMO-CO method proposed in this paper is effective, if there exist positive matrix $\bar{P}$ and matrices $\bar{N}_1$, $\bar{N}_2$ with appropriate dimensions that
\begin{equation}
\Xi =
\left[ \begin{array}{cc}
\Xi_{11}   & \Xi_{12}  \\
\ast       &  \Xi_{22}
\end{array} \right] <0
\label{Equ_Inequality_CSMOErr}
\end{equation}
where
\begin{eqnarray*}
\Xi_{11} &=& \bar{P} \bar{E}^{-1}\bar{A} +(\bar{P} \bar{E}^{-1}\bar{A})^T   - \bar{N}_1\bar{C} -\bar{C}^T \bar{N}_1^T, \\
\Xi_{12} &=& (\bar{N}_1 \bar{C})^T ,\\
\Xi_{22} &=&  \bar{P} \bar{E}^{-1}\bar{A} +(\bar{P} \bar{E}^{-1}\bar{A})^T  -\bar{N}_2 - \bar{N}_2^T.
\end{eqnarray*}
and the observer gains $\bar{L}= \bar{E}\bar{P}^{-1} \bar{N}_1$ and $\bar{\mathcal{L}} = \bar{E}\bar{P}^{-1}\bar{N}_2$.
\end{theorem}

\begin{proof}
Define the Lyapunov functional below
\begin{equation}
\begin{array}{rcl}
V(t)   &=& V_1(t)+V_2(t) \\
V_1(t) &=& \bar{e}^T(t) \bar{P}  \bar{e}(t) ,
V_2(t) =\bar{\epsilon}^T(t) \bar{P}  \bar{\epsilon} (t)
\end{array}
\end{equation}
the derivatives of $V_1(t)$ and $V_2(t)$ are
\begin{eqnarray*}
\dot{V}_1(t) &=& 2\bar{e}^T(t) \bar{P} \dot{\bar{e}}(t) \\
&=& 2\bar{e}^T(t) \bar{P} \bar{E}^{-1} \left(\bar{A}  - \bar{L}\bar{C}\right)\bar{e}(t) \\
&& +2\bar{e}^T(t) \bar{P} \bar{E}^{-1}  \bar{L} {C}_{\omega} (u_{s2}(t)-\omega)  \\
&& + 2\bar{e}^T(t) \bar{P} \bar{E}^{-1} \bar{B}_f\left(\bar{d}  - u_{s1}(t)\right)   \\
\dot{V}_2(t) &=& 2\bar{\epsilon}^T(t) \bar{P}  \dot{\bar{\epsilon}}(t) \\
&=& 2\bar{\epsilon}^T(t) \bar{P}  \bar{E}^{-1}\bar{A} \bar{\epsilon}(t)
- 2\bar{\epsilon}^T(t) \bar{P}  \bar{E}^{-1} \bar{\mathcal{L}} \bar{\epsilon}(t) \\
&&+2\bar{\epsilon}^T(t) \bar{P}  \bar{E}^{-1} \bar{L} \bar{C} \bar{e}(t) \\
&&+2\bar{\epsilon}^T(t) \bar{P}  \bar{E}^{-1} \bar{L} C_{\omega} \left( \omega(t)- u_{s3}(t) \right)
\end{eqnarray*}
With $\omega \in [-\bar{\omega}, \bar{\omega}]$, constraints $(H_1\bar{C})^T =$ $ \bar{P}\bar{E}\bar{B}_f$, $(H_2\bar{C})^T = $ $ \bar{P}\bar{E}^{-1}\bar{L}C_{\omega} $$= \bar{N}_2 C_{\omega}$, $ H_3^T = \bar{P}\bar{E}^{-1}\bar{L}C_{\omega} = \bar{N}_1 C_{\omega}$, and the discontinuous functions in \eqref{Equ_HC_Constraint}, similar in \textbf{Theorem}.\ref{Theo_SMO}, it verifies
$\bar{e}^T(t) \bar{P}  \bar{E}^{-1} \bar{B}_f(\bar{d} (t)- u_{s1}(t))  \leq   0$,
$\bar{e}^T(t) \bar{P} \bar{E}^{-1}\bar{L} {C}_{\omega} (u_{s2}(t)-\omega)  \leq  0$
and
\begin{eqnarray*}
&&   \bar{\epsilon}^T(t) \bar{P}  \bar{E}^{-1} \bar{L} C_{\omega} \left( \omega(t)- u_{s3}(t) \right) \\
&=&  s^T(t) \left( \omega(t)- \bar{\omega} \textbf{sgn}(s(t)) \right) \\
&\leq& (- \bar{\omega} +\omega(t)) |s^T(t)  | \\
&\leq& 0
\end{eqnarray*}
And one obtains
\begin{equation}
\begin{array}{rcl}
\dot{V}_1(t) &\leq& 2\bar{e}^T(t) \bar{P} \bar{E}^{-1}\left(\bar{A}  - \bar{L}\bar{C}\right)\bar{e}(t) \\
\dot{V}_2(t) &\leq& 2\bar{\epsilon}^T(t) \bar{P}  \bar{E}^{-1}\bar{A} \bar{\epsilon}(t)
- 2\bar{\epsilon}^T(t) \bar{P}  \bar{E}^{-1} \bar{\mathcal{L}} \bar{\epsilon}(t) \\
&&+2\bar{\epsilon}^T(t) \bar{P}  \bar{E}^{-1} \bar{L} \bar{C} \bar{e}(t)
\end{array}
\end{equation}
With $\bar{L}= \bar{E}\bar{P}^{-1} \bar{N}_1$ and $\bar{\mathcal{L}} = \bar{E}\bar{P}^{-1}\bar{N}_2$,
defining $\xi(t)=[\bar{e}^T(t), \bar{\epsilon}^T(t)]^T$,
it means that if the inequality \eqref{Equ_Inequality_CSMOErr} holds,
\begin{equation}
\dot{V}(t) = \xi^T(t) \Xi \xi(t) <0
\end{equation}
and \eqref{Equ_CSMOErr} is stable.
\end{proof}

\begin{remark}
The keys to designing the SMO in \eqref{Equ_Aug_SMO} and the SMO-CO in \eqref{Equ_Cascade_Observer} are the multiple discontinuous functional $u_{s1}$, $u_{s2}$ and $u_{s3}$.
The multiple discontinuous functions can make up the affection of the lumped disturbance and the sensor noises for the two layers of the observer proposed in this letter.
In addition to the cascade observation structure, this is another difference compared with other SMO methods.
\end{remark}

\subsection{SMO-CO based Compensator}
The SMO-CO based compensation controller is
\begin{equation}
u(t) = \bar{K} \tilde{\bar{x}}(t)
\label{Equ_SMOCO_Controller}
\end{equation}
where
$\tilde{\bar{x}}(t) = \hat{\bar{x}}(t) -  \bar{\epsilon} (t) = \bar{x}(t)- \bar{e}(t) - \bar{\epsilon}(t)$,
$\bar{K}$ is same as in \eqref{Equ_SMO_Controller} that $\bar{K} = \left[K ~ -B^{\dagger}B_f \right]$, and $B{\dagger}$ makes $BB{\dagger}B_f=B_f$.

The SMO-CO based closed-loop system can be verified
\begin{equation}
\left\{
\begin{array}{rcl}
\dot{x}  &=&  \left( A  + BK \right) x -B \bar{K} \bar{e} - B \bar{K} \bar{\epsilon} \\
\bar{E} \dot{\bar{e}}  &=&  \left(\bar{A}  - \bar{L}\bar{C}\right)\bar{e}  +\bar{L} {C}_{\omega} (u_{s2}(t)-\omega)  \\
&&+ \bar{B}_f(\bar{d}  - u_{s1}(t)) \\
\bar{E} \dot{\bar{\epsilon}}&=&\left( \bar{A}  -\bar{\mathcal{L}}\right)\bar{\epsilon}+ \bar{L}\bar{C}\bar{e}+ \bar{L}C_{\omega}(\omega(t) -  u_{s3}(t)) \\
\end{array}
\right.
\label{Equ_SMOCO_Controller}
\end{equation}

\begin{theorem} \label{Theo_SMO_CO_Controller}
The closed-loop system \eqref{Equ_SMOCO_Controller} is stable, i.e. the SMO-CO  \eqref{Equ_Aug_SMO} and \eqref{Equ_Cascade_Observer}, and the compensation feedback controller \eqref{Equ_SMOCO_Controller} are effective, when there exist positive matrices $Q$ and $\bar{P}$, and matrices $\bar{N}_1$, $\bar{N}_2$ with appropriate dimensions that
\begin{equation}
\Xi =\left[\begin{array}{ccc}
       \Xi_{11}& \Xi_{12}& \Xi_{13} \\
       \ast    & \Xi_{22}& \Xi_{23} \\
       \ast    & \ast    & \Xi_{33} \\
      \end{array} \right]
<0
\label{Equ_ClsdLp_Xi_S}
\end{equation}
where
\begin{eqnarray*}
&\Xi_{11}& = Q(A+BK)+(A+BK)^TQ, \\
&\Xi_{12}& = \Xi_{13} = -QB\bar{K},    \\
&\Xi_{22}& = \bar{P}\bar{E}^{-1}\bar{A} +(\bar{P}\bar{E}^{-1}\bar{A} )^T - \bar{N}_1 \bar{C} -\bar{C}^T \bar{N}_1^T, \\
&\Xi_{23}& = (\bar{N}_1 \bar{C})^T, \\
&\Xi_{33}& = \bar{P} \bar{E}^{-1}\bar{A} +(\bar{P} \bar{E}^{-1}\bar{A})^T  -\bar{N}_2 - \bar{N}_2^T.
\end{eqnarray*}
the observer gains can be selected as $\bar{L}= \bar{E}\bar{P}^{-1} \bar{N}_1$ and $\bar{\mathcal{L}} = \bar{E}\bar{P}^{-1}\bar{N}_2$.
\end{theorem}

The effectiveness of the closed-loop system \eqref{Equ_SMOCO_Controller} can be proved via similar methods in \textbf{Theorem \ref{Theorem_SMO_Compensator}} and \textbf{Theorem \ref{Theorem_SOM_CO}} by choosing a Lyapunov function below
\begin{equation}
V(t)= x^T(t)Qx(t)+ \bar{e}^T(t)\bar{P}\bar{e}(t)+ \bar{\epsilon}^T(t)\bar{P}\bar{\epsilon}(t)
\end{equation}

\begin{remark}
How to select the observer gains $\bar{L}$, $\mathcal{L}$ and design the discontinuous functionals $u_{s1}(t)$, $u_{s2}(t)$ and $u_{s}$ are the main work of this paper that only the sufficient condition on examining the effectiveness of the closed-loop system is provided for SMO and SMO-CO.
One can obtain the gain $K$ by left and right multiplying $\Xi$ in \eqref{Equ_ClsdLp_Xi} or \eqref{Equ_ClsdLp_Xi_S} with $diag(Q^{-1},I)$, or giving a set of desired poles to $(A+BK)$.
\end{remark}

\section{Example} \label{sec:Example}
In this section, to prove the effectiveness of the proposed method, one numerical example is provided.
The system model parameters in \eqref{Equ_SMO_Original_System} are
\begin{eqnarray*}
&A = \left[\begin{array}{cccc} 0 &0& 1 &0 \\
    0 &0 &0 &1\\
    0 &0.8 &-1.5 &0\\
    -3.7& 0.7 &0 &-4.9 \end{array} \right],
B = \left[\begin{array}{cc}0  & 0 \\
    0 &  0\\
   -2 &  0\\
    0 & 2.5 \end{array} \right], \\
& C_{\omega}= \left[\begin{array}{cc} 2& 0 \\ 0 &1\end{array} \right] ,
\Lambda = \left[\begin{array}{cc} 2& 0 \\ 0 &0.5 \end{array} \right] ,
C =\left[\begin{array}{cccc} 1& 0 &0 & 0\\ 0 &1 &0 & 0\end{array} \right].
\end{eqnarray*}
The initial system state is assumed to be $x_{t=0} =[-200, $ $ -100,$ $ 80, 60]^T$.
The parameter $A$ has positive eigenvalues.
Select the parameters $\bar{A}, \bar{B}, \bar{B}_f, \bar{E}, \bar{C}$ as in \eqref{Equ_SMO_Augment_System} and $\Phi = diag( 0.1,0.01) $.
The pole region for the observer is limited in the LMI region $(D(\pi/3, 20, -10))$ and $(D(\pi/3, 10, -6))$ for the gains $L$ and $\mathcal{L}$ respectively as in \cite{chilali1999robust}. The sum value $(\overline{\mathbf{d}}+ \overline{\mathbf{h}} \Phi^{-1}+\eta)$ in \eqref{Equ_Discontinus_Function} is set to be 1000. The sign function $\textbf{sgn}\left(s_i(t) \right)$ in \eqref{Equ_Discontinus_Function} and \eqref{Equ_Discontinus_Function_S}, is computed as
\begin{equation*}
 \textbf{sgn}\left(s_i(t) \right) = \cfrac{s_i(t)}{\|s_i(t) \|_{\infty}+\varsigma}
\end{equation*}
where $\varsigma = 0.01$.

The observer gains $L$, discontinuous functions $u_{si}$, for $i=1,2$, are chosen to be the same as in \eqref{Equ_Aug_SMO} and \eqref{Equ_Cascade_Observer}.
The controller gains $K$ in \eqref{Equ_SMO_Controller} and \eqref{Equ_SMOCO_Controller} are also selected to be the same.
The gain $K$ of $\bar{K}$  is

\begin{equation*}
K = \left[ \begin{array}{cccc}
 -135.0000  & -0.4000 & -15.7500  &       0 \\
   -1.4800  & 43.4800 &        0  &  6.4400
\end{array}\right]
\end{equation*}

The observer gain $L$ for SMO and SMO-CO is
\begin{equation*}
L = \left[ \begin{array}{cccc}
28.7646 &   1.6434 \\
-0.0599 &  25.4622\\
22317   &   4062.2\\
-1224.4  & 225890\\
-548.4370 & -100.7284\\
-9.8     &   1804.7\\
\end{array}\right]
\end{equation*}

The observe gain $\mathcal{L}$ for SMO-CO is
\begin{equation*}
\mathcal{L} = \left[ \begin{array}{cccccc}
8.6968  &  0.0005   & 0.9997  &  0.0000  & -0.0001  & -0.0000\\
0.0001  &  8.6854   & 0.0000  &  1.0001  &  0.0000  &  0.0000\\
0.5435  &  2.1518   & 7.2862  & -0.0619  &-347.8314 &   0.0113 \\
-2.9929 &  4.3330   & 0.0063  & 3.6058   &-0.0192   & 1087 \\
-0.0137 &  -0.0336  & -0.0023 &   0.0015 &   8.5957 &  -0.0003 \\
0.0057  &  0.0292   & 0.0000  & -0.0015  & -0.0002  &  8.6858 \\
\end{array}\right]
\end{equation*}

The gains $H_1$ and $H_2$ for discontinuous functions $u_{s1}$ and $u_{s2}$ are
\begin{equation*}
\begin{array}{rcl}
H_1 &=& \left[ \begin{array}{cc}
 -6.0100*10^{-9}  &   7.1104*10^{-10 }   \\
7.1104*10^{-10 } &  2.2234*10^{-9  }
\end{array}\right], \\ \\
H_2 &=& \left[ \begin{array}{cc}
1.9193*10^{-4}    &   -7.7971*10^{-6 }   \\
 -7.7971*10^{-6 } &  8.4519*10^{-4 }
\end{array}\right].
\end{array}
\end{equation*}
%%%%%%%%%%%%%%%%%%%%%%%%%%%%%%%%%%%%%%%%%%%%%%%%%%%%%%%%%%%%%%%%%%%%%%%%%%

The gain $H_3$ for discontinuous function $u_{s3}$ is
\begin{equation*}
H_3^T = \left[ \begin{array}{cccccc}
1.9193*10^{-4}  & -7.7907*10^{-6}\\
-7.8048*10^{-6} & 8.4519*10^{-4}\\
-1.4183*10^{-6} & 1.2944*10^{-7} \\
-7.1583*10^{-8}  &-3.8856*10^{-7}\\
2.1464 *10^{-6} &  1.0194*10^{-8}\\
-1.3308*10^{-7}&  -2.5957*10^{-6}\\
\end{array}\right]
\end{equation*}

The estimated disturbance by the SMO and SMO-CO are shown in Fig.\ref{Fig_EstDis_SMOCO}. It shows the SMO and SMO-CO methods are both available for the estimation of disturbances in sinusoidal and step forms.
Connecting the estimated disturbance $\hat{d}(t)$ with a low pass filter, one obtains the filtered estimated disturbance $d_{f}(t)$, i.e. the SMO-LF based estimated disturbance as below
\begin{equation*}
d_{f}(t) = \cfrac{1}{0.01s+1} \hat{d}(t)
\end{equation*}
to compare with the SMO-CO method.
Presented in the left and right sub-figures of Fig.\ref{Fig_EstDis_SMOCO}, during the beginning and the middle periods, the estimated disturbance under SMO-CO is with lower amplitude than the SMO and the SMO-LF methods.
\begin{figure}[hp]
\centering
	\includegraphics[ width=0.8 \linewidth]{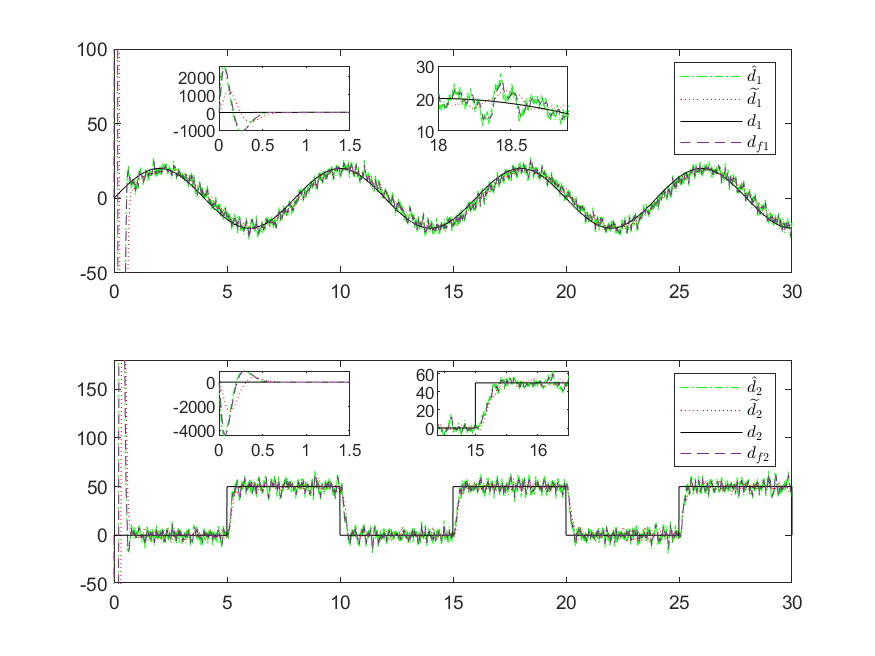}
	\caption{The Real Disturbance $d$, Estimated Disturbances by SMO-CO $\widetilde{d}$, SMO $\hat{d}$ and SMO-LF $d_{f}$  }
    \label{Fig_EstDis_SMOCO}
\end{figure}

Shown in Fig.\ref{Fig_EstErrors_SMOCO}, the observation errors for $x_1-\hat{x}_1$ and $x_1-\widetilde{x}_1$ are with similar amplitude.
But for the other three system states, the observation errors between the real and the SMO-CO based estimated states are lower than the SMO based ones.
\begin{figure}[hp]
\centering
	\includegraphics[ width=0.8 \linewidth]{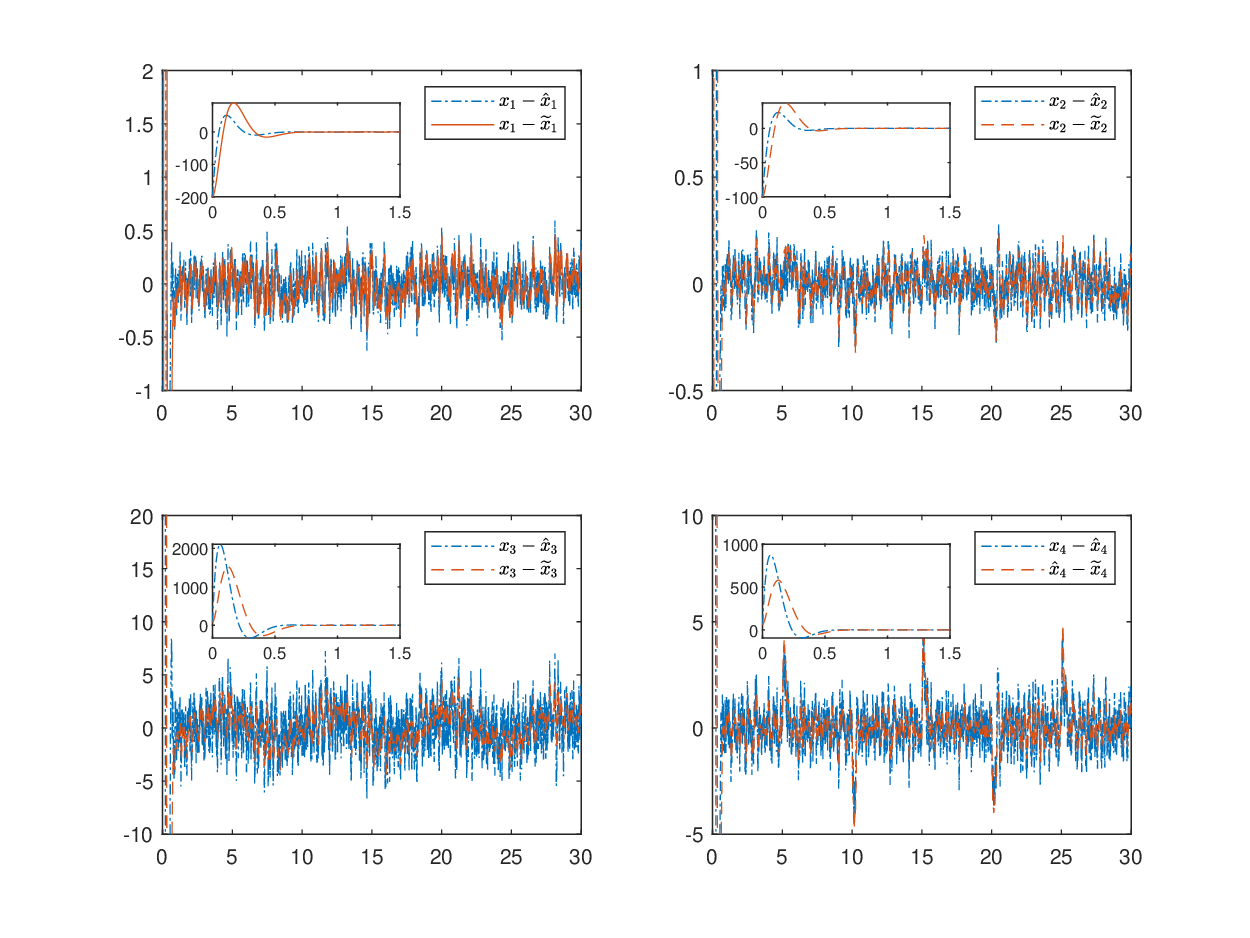}
	\caption{Estimation Errors under SMO and SMO-CO}
    \label{Fig_EstErrors_SMOCO}
\end{figure}

The interesting findings about the control inputs under SMO and SMO-CO are depicted in Fig.\ref{Fig_CtrInputs_SMOCO}.
In the sub-figures of Fig.\ref{Fig_CtrInputs_SMOCO}, at the initial stage, the differences between the control inputs under SMO and SMO-CO are not apparent.
While, after the initial 0.5 seconds, the control under SMO-CO is much smoother and with lowerer amplitude than one under SMO.
It means, though with the same controller gain $K$, the control input with SMO-CO consumes less energy.
\begin{figure}[hp]
\centering
	\includegraphics[ width=0.8 \linewidth]{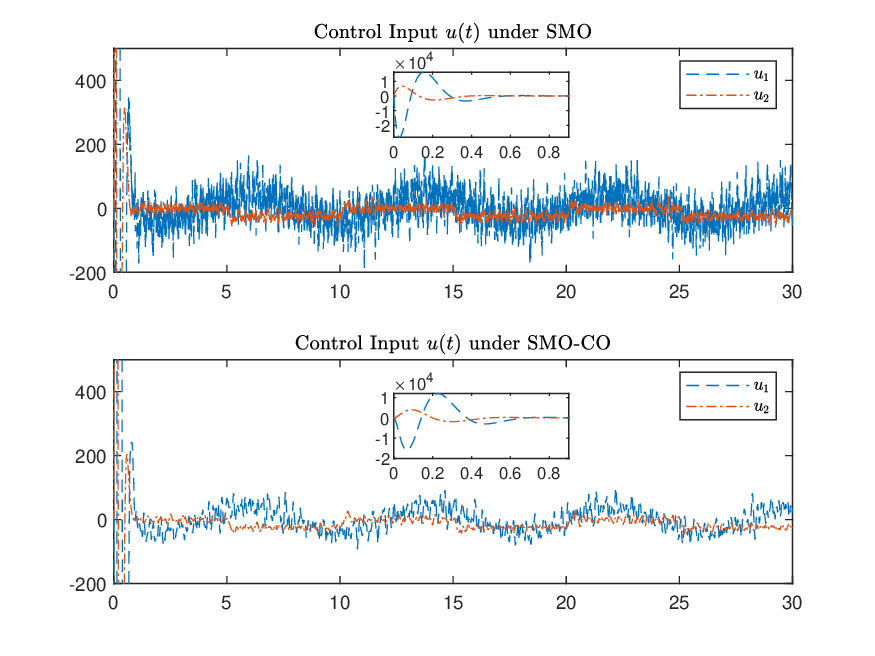}
	\caption{Control Inputs under SMO and SMO-CO}
    \label{Fig_CtrInputs_SMOCO}
\end{figure}

%%%%%%%%%%%%%%%%%%%%%%%%%%%%%%%%%%%%%%%%%%%%%%%%%%%%%%%%%%%%%%%%%

\begin{table}[hp]
\begin{center}
\caption{2-Norm Values Performance Comparison for Estimation Error and Inputs\label{table}}
\begin{tabular}{ cc|cc }
  \hline
  $\|x - \hat{x}\|_2$ & $\|x-\widetilde{x}\|_2$ & $\|u(\hat{x},\hat{d})\|_2$  & $\|u(\widetilde{x},\widetilde{d})\|_2$    \\ \hline
   2.17 & 1.83  &  53.22 & 40.49\\
  \hline
\end{tabular}
\end{center}
\end{table}

Table.\ref{table} presents the 2-norm values of the observation errors of the system state and control consumption during a one-round simulation, for both traditional SMO and SMO-CO methods, over the time interval of 1s to 30s. The SMO-CO based control scheme proposed in this paper has lower observer error and requires less control input energy.

\section{Conclusion} \label{sec:Conclusion}
A new sliding mode based cascade observer scheme, which is composed of the SMO and CO layers, is proposed in this paper.
Multiple discontinuous functionals are designed in the observers to improve estimation quality.
The observer based compensation controller is also offered.
An alternative observer design method, along with a sufficient condition for examining the effectiveness of the closed-loop system, is provided.
The example shows that, compared with the conventional SMO method, the SMO-CO in this paper can further improve the disturbance and system state estimation quality.
Interestingly, with the same feedback gains and compensation method, the control consumption with the SMO-CO is less than that with traditional SMO.

The validity of the SMO-CO is demonstrated through a numerical example. The boundary and Lipschitz continuity limits for the lumped disturbance, the approximation computations for  $\bar{C}\bar{e}(t)$ in the sign function, and the constraints \eqref{Equ_HC_Constraint} and \eqref{Equ_HC_Constraint_S} are required for SMO and SMO-CO in this letter.
%%%%%%%%%%%%%%%%%%%%%%%%%%%%%%%%%%%%%%%%%%%%%%%%%%%%%%%%%%%%%%%%%%%%%%%%%%
\section*{References}
%%%%%%%%%%%%%%%%%%%%%%%%%%%%%%%%%%%%%%%%%%%%%%%%%%%%%%%%%%%%%%%%%%%%%%%%%%%%
\bibliographystyle{plain}
\bibliography{Ref}

\end{document}